\documentclass[11pt,a4paper]{article}
\usepackage{t1enc}
\usepackage[latin1]{inputenc}
\usepackage[english]{babel}
\normalfont
\usepackage{amsmath}
\usepackage{amsthm}
\usepackage{yfonts}
\usepackage{bbm}
\usepackage{bm}
\usepackage{wasysym}
\usepackage{amssymb}
\usepackage{mathrsfs}
\usepackage{pifont}
\usepackage{hyperref}

\newcommand{\be}[0]{\begin{equation}}
\newcommand{\ee}[0]{\end{equation}}
\newcommand{\C}[0]{\mathbbm{C}}

\newcommand{\N}[0]{\mathbbm{N}}
\newcommand{\R}[0]{\mathbbm{R}}
\newcommand{\Z}[0]{\mathbbm{Z}}
\newcommand{\hs}[0]{{\rm hs}}

\theoremstyle{plain}
\newtheorem{theorem}{Theorem}[section]
\newtheorem{lemma}[theorem]{Lemma}

\theoremstyle{definition}

\numberwithin{equation}{section}

\begin{document}

\vspace*{-1cm}
\thispagestyle{empty}
\vspace*{1.5cm}

\begin{center}
{\Large 
{\bf Finite higher spin transformations from exponentiation}}
\vspace{2.0cm}

{\large Samuel Monnier}
\vspace*{0.5cm}

Institut für Mathematik, 
Universität Zürich,\\
Winterthurerstrasse 190, 8057 Zürich, Switzerland

\vspace*{1cm}

{\bf Abstract}
\end{center}

We study the exponentiation of elements of the gauge Lie algebras $\hs(\lambda)$ of three-dimensional higher spin theories. Exponentiable elements generate one-parameter groups of finite higher spin symmetries. We show that elements of $\hs(\lambda)$ in a dense set are exponentiable, when pictured in certain representations of $\hs(\lambda)$, induced from representations of $SL(2,\mathbbm{R})$ in the complementary series. We also provide a geometric picture of higher spin gauge transformations clarifying the physical origin of these representations. 
This allows us to construct an infinite-dimensional topological group $HS(\lambda)$ of finite higher spin symmetries. Interestingly, this construction is possible only for $0 \leq \lambda \leq 1$, which are the values for which the higher spin theory is believed to be unitary and for which the Gaberdiel-Gopakumar duality holds. We exponentiate explicitly various commutative subalgebras of $\hs(\lambda)$. Among those, we identify families of elements of $\hs(\lambda)$ exponentiating to the unit of $HS(\lambda)$, generalizing the logarithms of the holonomies of BTZ black hole connections. Our techniques are generalizable to the Lie algebras relevant to higher spin theories in dimensions above three.

\newpage

\tableofcontents

\section{Introduction and summary}

It is believed that the perturbative formulation of string theory hides a large part of the symmetries of the underlying fundamental theory. The string tension should be an order parameter for the breaking of the hidden symmetry: as the tension goes to zero, the tower of massive string modes becomes massless, resulting in the restoration of a large gauge symmetry. The modes becoming massless are associated to field theories describing particles of spin higher than 2 (henceforth \textit{higher spins}), whose formulation even at the classical level is challenging.

While the precise effective field theory describing the zero tension limit of string theory, if it exists at all, has remained elusive, much progress has been made by Vasiliev and others in formulating field theories involving higher spins in various dimensions (see for instance \cite{Vasiliev:1992gr, Vasiliev:1996hn, Vasiliev:1995dn, Bekaert:2005vh} for papers relevant to the three-dimensional case of interest to us here). One can therefore entertain the hope that understanding their symmetries could provide a glimpse of the hidden symmetry group of string theory. Moreover many of these theories, formulated on anti de Sitter space, seem to admit quantum field theories as holographic duals, suggesting that they are the classical limits of consistent theories of quantum gravity.

Like gravity, higher spin theories admit first order formulations involving connections valued in certain Lie algebras, which encode the infinitesimal higher spin symmetries. These Lie algebras are well-understood: they are enveloping algebras of non-compact real forms of semi-simple Lie algebras, quotiented by a certain ideal \cite{Iazeolla:2008ix, Bekaert:2008sa, Boulanger:2011se, Joung:2014qya}. The corresponding groups of finite higher spin symmetries have to our knowledge not yet been constructed. The aim of this paper is to construct and begin to study the simplest of them.

We will therefore focus on the higher spin Lie algebras $\hs(\lambda)$ underlying certain higher spin theories in $AdS_3$, well known to be holographically dual to large $N$ limits of exactly solvable two-dimensional conformal field theories \cite{Gaberdiel:2010pz, Gaberdiel:2012uj}. 
$\hs(\lambda)$ is essentially the enveloping algebra $U({\rm sl}(2,\mathbbm{R}))$, seen as a Lie algebra and quotiented by the ideal $\Omega - \frac{1}{2}(\lambda^2-1) \mathbbm{1}$, where $\Omega$ denotes the quadratic Casimir of ${\rm sl}(2,\mathbbm{R})$. $\lambda$ can a priori be any complex number, and for $\lambda$ an integer larger than 1, $\hs(\lambda)$ admits a quotient isomorphic to ${\rm sl}(\lambda,\mathbbm{R})$, associated with higher spin theories with a finite number of higher spin fields \cite{oai:arXiv.org:1008.4744}. However, the holographic picture suggests that the associated higher spin theories are unitary only for $\lambda$ real in the window $0 \leq \lambda \leq 1$, and we will obtain further independent evidence for this.

Naively, given a Lie algebra $\mathfrak{g}$ which, like $hs(\lambda)$, is also an associative algebra, one may want to define the corresponding Lie group as the group generated by the series
\be
\label{EqFormPowSerExp}
\exp x = \sum_{n = 0}^\infty \frac{1}{n!} x^n \;, \quad x \in \mathfrak{g} \;.
\ee
This works perfectly well for finite-dimensional algebras, or even for infinite-dimensional Banach algebras, but as we explain in Section \ref{SecExpAdjRep}, this fails for $\hs(\lambda)$. The problem is that in order to make sense of the formal series \eqref{EqFormPowSerExp} as an actual infinite sum, we need a norm, or at least a topology on $\mathfrak{g}$ ensuring the convergence of \eqref{EqFormPowSerExp}. We present a simple argument that no such norm exists for $\hs(\lambda)$.

There are however other ways of exponentiating a Lie algebra. For instance, given a faithful representation, one may hope to be able to exponentiate the operators representing the Lie algebra elements.\footnote{This approach has been mentioned before in \cite{Campoleoni:2013lma, Gaberdiel:2013jca}, but the representations of ${\rm sl}(2,\mathbbm{R})$ considered there do not integrate to representations of $SL(2,\mathbbm{R})$.} In particular, if the representation is unitary, one may hope that the Lie algebra elements are represented by skew-adjoint operators, which do exponentiate to one-parameter groups of unitary transformations by the spectral theorem. Given that $\hs(\lambda)$ admits an $sl(2,\mathbbm{R})$ subalgebra (related to local Poincaré transformations in the higher spin theory), a zero order requirement for such a representation is that it integrates to a representation of $SL(2,\mathbbm{R})$. It turns out that there exists a unique unitary representation of $SL(2,\mathbbm{R})$ inducing a faithful representation of $\hs(\lambda)$ for each $\lambda$ in the window $0 < \lambda < 1$.\footnote{There are as well unitary representations of $\hs(\lambda)$ for $\lambda = 0,1$, induced from representations of $SL(2,\R)$ in the principal and discrete series, respectively. We will not consider them in this work, but we expect that the same ideas would apply to these limit cases.} These are infinite-dimensional representations $\mathcal{C}_\lambda$ of $SL(2,\mathbbm{R})$ in the complementary series.

In order to justify the use of the complementary series of representation, we offer a geometric interpretation of the higher spin gauge symmetry of the three-dimensional higher spin theory, based on ideas in \cite{Vasiliev:1996hn, Bekaert:2008sa}. In a nutshell, the higher spin fields are gauging the local symmetries of the Klein-Gordon equation governing the scalar field of the theory. We show that representations in the complementary series are realized in this context, with the correct parameter $\lambda$, by the action of the local Poincaré isometries on the solutions of the Klein-Gordon equation. We believe these ideas should eventually provide a conceptual definition of the higher spin symmetry group, but we leave this question for future work.

As the representations $\mathcal{C}_\lambda$ are unitary, the elements of (the relevant real form of) $\hs(\lambda)$ are represented by skew-symmetric operators, and representation theory \cite{Nelson1959} provides sufficient conditions for these operators to be skew-adjoint. Somewhat unfortunately, not all of these operators are skew-adjoint and exponentiate. The set of exponentiable elements of $\hs(\lambda)$ is however dense, in the sense that one can always add an arbitrary small perturbation to any element of $\hs(\lambda)$ in order to make it exponentiable. We can then define the higher spin group $HS(\lambda)$ as the group generated by the one-parameter groups of unitary transformations associated to exponentiable elements. $HS(\lambda)$ inherits a natural group topology, but, at least in this topology, it is not a Lie group (i.e. it is not an infinite-dimensional smooth manifold).

Let us emphasize that $HS(\lambda)$ is \textit{not} the higher spin analogue of the diffeomorphism group. Rather, $HS(\lambda) \times HS(\lambda)$ (or one of its double covers) is the structure group of the principal bundle generalizing the frame bundle in the first order formalism for classical 3d gravity, i.e. it generalizes the $(SL(2,\mathbbm{R}) \times SL(2,\mathbbm{R}))/\mathbbm{Z}_2$ group of local Poincaré symmetries of gravity in $AdS_3$. It is well-known that the diffeomorphism group can be recovered from field-dependent local Poincaré transformations \cite{Witten:1988hc}. Similarly, we expect the higher spin symmetry group of spacetime, generalizing the diffeomorphism group, to be realized as field-dependent $HS(\lambda) \times HS(\lambda)$-transformations.

The representations $\mathcal{C}_\lambda$ admit explicit models in terms of differential operators on the real line or on the unit circle, which allows us to exponentiate explicitly several maximal commutative subalgebras of $\hs(\lambda)$ by diagonalizing the corresponding differential operators. This gives us a first glimpse into the global features of $HS(\lambda)$: we describe two (conjugacy classes of) commutative subgroups isomorphic to $\mathbbm{R}^\mathbbm{N}$, the countable direct product of copies of the additive group $\mathbbm{R}$, and one commutative subgroup isomorphic to $U(1)^\mathbbm{N}$. 

Associated to the latter are elements of $\hs(\lambda)$ exponentiating to the unit of $HS(\lambda)$. Such elements are of great interest to the construction of spherically symmetric solutions of the corresponding higher spin theories (see for instance \cite{Gutperle:2011kf, Kraus:2012uf}). We recover the logarithm of the holonomy of the BTZ black hole connection along the time-like circle as one of these elements, and find natural generalizations. We hope that these results will be useful to construct new solutions of the higher spin theory, in particular black holes with finite higher spin charges. (The current construction in \cite{Kraus:2012uf} is perturbative in the higher spin charge, and the convergence of the associated series is not clear.)

The techniques used in this paper are in principle applicable to any Lie algebra that can be pictured as the quotient of an enveloping algebra of a semi-simple Lie algebra by the annihilator of a unitary representation. Several works suggest that many higher spin Lie algebras associated with theories in higher dimensions are of this form \cite{Bekaert:2008sa, 2011arXiv1107.5840M,Joung:2014qya}. Technical difficulties might appear at the computational level, as the ordinary linear differential equations appearing throughout this paper will in general turn into partial differential equations.

Other interesting points are worth mentioning. First, as remarked above, the complementary series representations provide two models of $\hs(\lambda)$ in terms of differential operators in one variable. These models are different from the one appearing in the higher spin literature in the guise of the star product representation. They could be used to construct new star product representations, which might prove computationally useful in certain situations. 

Although we do not attempt this here, our construction makes possible the identification of the $SL(2,\mathbbm{R})$ subgroups of $HS(\lambda)$. We expect each of these subgroups to be associated with an $AdS_3$ vacuum of the higher spin theory, coming with a spectrum of higher spin fields determined by the embedding. We emphasize that in this infinite-dimensional setting, it is not sufficient to study this problem at the level of Lie algebras, as there might very well be ${\rm sl}(2,\mathbbm{R})$ subalgebras in $\hs(\lambda)$ which do not exponentiate.\\

The paper is organized as follows. Section \ref{SecHigSpinLieAlg} is a review of the construction of the higher spin Lie algebras of interest to us. In Section \ref{SecExpAdjRep}, we review old results of Dixmier, who characterized the elements of $\hs(\lambda)$ whose exponential have a well-defined adjoint action, mapping finite sums of generators to finite sums of generators (Section \ref{SecClassExpElAdjRep}). We also review his characterization of the inner automorphism group of $\hs(\lambda)$ (Section \ref{SecGrpInnAut}). We provide explicit examples of Lie algebra elements which are not exponentiable in the sense above (Section \ref{SecNonExpElAdj}) and explain the problems encountered when trying to make sense of the formal series \eqref{EqFormPowSerExp} (Section \ref{SecCompl}). In Section \ref{SecExpCompSeries}, we turn to the problem of exponentiating $\hs(\lambda)$ in the complementary series of representations of $SL(2,\mathbbm{R})$. We start by reviewing basic facts about the represenation theory of $SL(2,\mathbbm{R})$ (Section \ref{SecBasFactRepTh}) and show how representations in the complementary series appear naturally in the higher spin theory (Section \ref{SecPhysMot}). After some technical results about representations in the complementary series (Section \ref{SecFaithCompSerRep} and \ref{SecCircModCompSerRep}), we provide the characterization of the exponentiable elements of $\hs(\lambda)$ (Section \ref{SecExpElComp}) and exhibit a non-exponentiable element (Section \ref{SecV31IsNotEssSA}). We define the higher spin group, its topology, and discuss why it is not a Lie group (Section \ref{SecPropGrHSSym}). We also describe the Euclidean higher spin group and mention some puzzles related to the adjoint action on $\hs(\lambda)$. In Section \ref{SecExFinHSTrans}, we consider three different maximal commutative subalgebras of $\hs(\lambda)$ and exponentiate them explicitly by diagonalizing the differential operators associated to their generators. We also characterize certain families of elements of $\hs(\lambda)$ exponentiating to the identity of $HS(\lambda)$. Appendix \ref{SecRevMathRes} reviews standard material about operators in Hilbert spaces, as well as less standard material about representations of Lie algebras by unbounded operators.

\section{The higher spin Lie algebra}

\label{SecHigSpinLieAlg}

We review here the construction of the complex higher spin Lie algebra, as well as the real form of interest to us.

\subsection{The complex higher spin Lie algebra}

\label{SecCompHSLieAlg}

Most of the material in this section can for instance be found in Section 2.1.3 of \cite{Gaberdiel:2012uj}, where the reader will also find pointers to the literature. We start by defining the complex version of $\hs(\lambda)$, $\hs_{\C}(\lambda)$. The Lie algebra $\mathfrak{sl}(2,\C)$ consists of complex $2\times 2$ matrices with vanishing trace. A convenient basis is
\be
X_+ = \left( \begin{array}{cc} 0 & 1 \\ 0 & 0 \end{array} \right) \;, \quad X_- = \left( \begin{array}{cc} 0 & 0 \\ 1 & 0 \end{array} \right) \;, \quad X_0 = \left( \begin{array}{cc} 1 & 0 \\ 0 & -1 \end{array} \right) \;.
\ee
These generators satisfy the commutation relations
\be
[X_0,X_\pm] = \pm 2 X_\pm \;, \quad [X_+, X_-] = X_0 \;.
\ee
The Casimir operator is
\be
\Omega = X_+ X_- + X_- X_+ + \frac{1}{2}X_0^2 \;.
\ee
(Remark that our definition of the Casimir differs by a factor 2 from the definition of Section 2.1.3 of \cite{Gaberdiel:2012uj}: $C^{\mathfrak{sl}} = \Omega/2$. This is the reason for an extra factor 2 in the definition of $B(\lambda)$ below.) 

Let $U(\mathfrak{sl}(2,\C))$ be the enveloping algebra of $\mathfrak{sl}(2,\C)$ and define $B(\lambda)$ to be the quotient of $U(\mathfrak{sl}(2,\C))$ by the ideal generated by $\Omega - \frac{1}{2}(\lambda^2-1) \mathbbm{1}$, for $\lambda \in \C$. We endow $B(\lambda)$ with the Lie bracket given by the commutator of the associative product in $U(\mathfrak{sl}(2,\C))$. We have the decomposition of Lie algebras
\be
\label{EqFactDirSumId}
B(\lambda) = \C \oplus \hs_\C(\lambda) \;,
\ee
which we take as the definition of the higher spin Lie algebra $\hs_\C(\lambda)$. Define
\be
\label{EqDefVSBasisHSL}
V^s_n = \frac{(n+s-1)!}{(2s-2)!} \big[ \underbrace{X_+, ... [X_+,[X_+}_{s-1-n \; {\rm times}},X_-^{s-1}]] \big]
\ee
for $s \in \N$, $s > 2$, $n \in \Z$, $|n| \leq s$. (This definition coincides with the one in Section 2.1.3 of \cite{Gaberdiel:2012uj}, although our choice of generators of $\mathfrak{sl}(2,\C)$ is slightly different.) $\{V^s_n\}$ is a basis of $\hs_\C(\lambda)$ as a vector space. More explicitly, we have for $s = 2,3$
$$
V^2_1 = X_- \;, \quad V^2_0 = \frac{1}{2}X_0 \;, \quad V^2_{-1} = -X_+ \;,
$$
\be
V^3_2 = X_-^2 \;, \quad V^3_1 = \frac{1}{2}(X_0X_-  + X_-) \;, \quad V^3_{0} = \frac{1}{4}X_0^2 - \frac{1}{6} \Omega = \frac{1}{4}X_0^2 - \frac{1}{12}(\lambda^2-1) \;,
\ee
$$
V^3_{-1} = -\frac{1}{2}(X_+ X_0 + X_+) \;, \quad V^3_{-2} = X_+^2\;.
$$
We write $\hs^{(s)}_\C(\lambda)$ for the subspace generated by $\{V^s_n\}$. We have the commutation relations
\be
\label{EqCommRelV2m}
[V^2_m,V^s_n] = (-n + m(s-1)) V^s_{m+n} \;,
\ee
showing that $\hs^{(s)}_\C(\lambda)$ forms a $2s-1$-dimensional irreducible module of the $\mathfrak{sl}(2,\C)$ subalgebra $\hs^{(2)}_\C(\lambda)$. There exists a closed formula for the commutation relations of the basis elements $\{V^s_n\}$ \cite{Pope:1989sr}, which we will not need here. Finally, $\hs^{(2)}_\C(\lambda) \oplus \hs^{(3)}_\C(\lambda)$ generates $\hs_\C(\lambda)$ as a Lie algebra.

\subsection{Real forms}

\label{SecRealForms}

We now consider real forms of $\hs_\C(\lambda)$ which restrict to $\mathfrak{sl}(2,\R)$ on $\hs^{(2)}_\C(\lambda)$. Recall that a real form of a complex Lie algebra $\mathfrak{g}$ is given by the fixed subalgebra of an antilinear involutive automorphism. Let $\theta_0$ be the antilinear automorphism of $\mathfrak{sl}(2,\C)$ leaving $\mathfrak{sl}(2,\R)$ fixed, i.e. complex conjugation. $\theta_0$ acts trivially on the real generators $X_+$, $X_0$, and $X_-$, and obviously extends to $\hs_{\C}(\lambda)$. Let $\theta_1$ be the involution of $\hs_{\C}(\lambda)$ acting by the identity on  $\hs^{(2)}_\C(\lambda)$ and by multiplication by $-1$ on $\hs^{(3)}_\C(\lambda)$. There are a priori two real forms of $\hs_\C(\lambda)$ of potential interest for higher spin theories, given by the fixed subalgebras of $\theta_0$ and $\theta_0 \theta_1$. In the following, we will only consider the latter and write it $\hs(\lambda)$. As a real Lie algebra, it is generated by $\{V^2_n, iV^3_m\}$ with $-1 \leq n \leq 1$ and $-2 \leq m \leq 2$, and was called ${\rm su}(\infty,\infty)$ in \cite{Pope:1989sr}. 
$\hs(\lambda)$ is the natural real form to consider for the following reason:\footnote{The physical relevance of $\hs(\lambda)$ is also supported by holography. Writing $W^k_0$ for the zero modes of the generators of the W-algebra of the CFT dual to the higher spin theory, equations (4.15) to (4.19) of \cite{Gaberdiel:2013jpa} show that the generators $W^{2n}_0$ and $iW^{2n+1}_0$ have a real spectrum in the CFT. The representation considered in \cite{Gaberdiel:2013jpa} is therefore a unitary representation of the real W-algebra generated by $W^{2n}(z)$ and $iW^{2n+1}(z)$. Taking the limit of infinite central charge and restricting to the wedge algebra \cite{Bowcock:1991zk,Gaberdiel:2011wb}, one recovers $\hs(\lambda)$ as a subalgebra of this real W-algebra. We thank Matthias Gaberdiel for pointing this out to us.} 
\begin{lemma}
\label{LemRepASAOp}
Let $\rho$ be an irreducible representation of $\mathfrak{sl}(2,\R)$ by skew-symmetric operators, with $\rho(\Omega) = \frac{1}{2}(\lambda^2-1)\mathbbm{1}$. Then $\rho$ extends to a representation of $\hs(\lambda)$ by skew-symmetric operators.
\end{lemma}
\begin{proof}
$\rho(V^2_n)$ are skew-symmetric by hypothesis. Let us consider the generator $iV^3_2 = iX_-^2$. $\rho(iV^3_2) = i\rho(X_-)^2$. $\rho(X_-)$ is skew-symmetric, so $\rho(X_-)^2$ is symmetric and $i\rho(X_-)^2$ is skew-symmetric. $iV^3_m$ can be obtained by the repeated action of $[X_+, .]$ on $iV^3_2$. As the commutator of two skew-symmetric operators is skew-symmetric, $\rho(iV^3_m)$ is skew-symmetric. All the Lie algebra generators of $\hs(\lambda)$ are represented by skew-symmetric operators, so this is true for all of $\hs(\lambda)$.
\end{proof}
Remark that Lemma \ref{LemRepASAOp} does not hold for the real form of $\hs^{(3)}_\C(\lambda)$ obtained from $\theta_0$.

\section{Exponentiation in the adjoint representation}

\label{SecExpAdjRep}

Just like three-dimensional gravity, three-dimensional higher spin theories can be seen as Chern-Simons theories. In the case of higher spin, the ${\rm sl}(2,\mathbbm{R}) \oplus {\rm sl}(2,\mathbbm{R})$ connection of three-dimensional gravity is replaced by a $\mathfrak{g} \oplus \mathfrak{g}$-valued connection, for some real Lie algebra $\mathfrak{g}$ admitting ${\rm sl}(2,\mathbbm{R})$ as a subalgebra \cite{oai:arXiv.org:1008.4744, oai:arXiv.org:1107.0290}. For the theories of interest to us, $\mathfrak{g} = \hs(\lambda)$. One of the reason to try to construct an exponential map for $\hs(\lambda)$ is therefore to understand finite gauge transformations of $\hs(\lambda)$-valued connections. Clearly, a hypothetic group exponentiating $\hs(\lambda)$ would act on the latter via the adjoint representation, by inner automorphisms. It is therefore a good idea to determine the group of inner automorphisms of $\hs(\lambda)$.

$\hs(\lambda)$ is an infinite-dimensional Lie algebra. Without introducing an extra structure, such as a topology or a norm, only finite linear combinations of generators are meaningful. As we explain in Section \ref{SecCompl}, we do not know any consistent completion of $\hs(\lambda)$ which would allow us to consider infinite linear combinations of generators. (The existence of such a completion will however be suggested by our results in Section \ref{SecExpCompSeries}). Consequently, in this section only finite sums of generators of $\hs(\lambda)$ are considered meaningful.

In this setting, the automorphisms of $B(\lambda)$ as an associative algebra were studied in an early paper by Dixmier \cite{Dixmier1973}. We will summarize here his results and deduce from them the inner automorphism group of $\hs(\lambda)$.

\subsection{Classification of the exponentiable elements in the adjoint representation}

\label{SecClassExpElAdjRep}

In our restricted setup where we allow only for finite linear combinations of generators, an element $X \in B(\lambda)$ is exponentiable only if all its adjoint orbits are finite-dimensional. In other words, let ${\rm ad}_X$ denote the adjoint action of $X$ on $B(\lambda)$: ${\rm ad}_X(Y) = [X,Y]$. Define $F(X)$ as the set of $Y \in B(\lambda)$ such that the span of $(ad_X)^n$, $n \in \mathbbm{N}$ is finite-dimensional. Then $X$ is exponentiable if and only if $F(X) = B(\lambda)$. Dixmier shows that there are exactly two types of exponentiable elements. 
\begin{enumerate}
	\item $X$ is a \textit{strictly semi-simple} element if ${\rm ad}_X$ can be diagonalized (by eigenvectors which are finite linear combinations of the generators). An example of a strictly semi-simple element is provided by $V^2_0$. Indeed, from \eqref{EqCommRelV2m} we have $[V^2_0, V^s_n] = -n V^s_n$.  
	\item $X$ is a \textit{strictly nilpotent} element if ${\rm ad}_X$ is locally nilpotent, i.e. if for each $Y \in B(\lambda)$, there is an $n \in \mathbbm{N}$ such that $({\rm ad}_X)^n(Y) = 0$. $V^2_1$ is a strictly nilpotent element, because $({\rm ad}_{V^2_1})^{s-n}(V^s_n) = 0$.
\end{enumerate}
If $X$ is strictly nilpotent, then any polynomial $P(X)$ in $X$ is strictly nilpotent as well. The example above shows that $V^s_{s-1} = (X_-)^{s-1}$ and $V^s_{-s+1} = (X_+)^{s-1}$ are all strictly nilpotent, and therefore exponentiate. We write
\be
\Phi_{n,\mu} = \exp {\rm ad}(\mu X_-^n) \;, \quad \Psi_{n,\mu} = \exp {\rm ad}(\mu X_+^n) \;.
\ee

\subsection{The group of inner automorphisms}

\label{SecGrpInnAut}

Dixmier's main results are summarized by the following theorem:
\begin{theorem}
(Dixmier, \cite{Dixmier1973}) All the (associative algebra) automorphisms of $B(\lambda)$ are inner and the group of automorphisms ${\rm Aut}(B(\lambda))$ is generated by the elements $\Phi_{n,\mu}$, $\Psi_{n,\mu}$, $n > 0$, $\mu \in \mathbbm{C}$.
\end{theorem} 
The group of inner Lie algebra automorphisms of $\hs_{\mathbbm{C}}(\lambda)$ coincides with ${\rm Aut}(B(\lambda))$, from which we obtain easily
\begin{theorem}
\label{ThGrInnAutHSL}
The group of inner automorphisms ${\rm Aut}_{\rm in}(\hs(\lambda))$ of $\hs(\lambda)$ is given by the real form of ${\rm Aut}(B(\lambda))$ left invariant by $\theta_0 \theta_1$ (see Section \ref{SecRealForms}). It is generated by $\Phi_{2n,\mu}$, $\Psi_{2n,\mu}$, $\Phi_{2n+1,i\mu}$, $\Psi_{2n+1,i\mu}$, $n >0$, $\mu \in \mathbbm{R}$.
\end{theorem} 
${\rm Aut}_{\rm in}(\hs(\lambda))$ contains $PSL(2,\mathbbm{R})$ (the adjoint group of $SL(2,\mathbbm{R})$) as a subgroup, as it should. 

\subsection{Non-exponentiable elements}

\label{SecNonExpElAdj}

Not all of the elements of $\hs(\lambda)$ exponentiate to an element of ${\rm Aut}_{\rm in}(\hs(\lambda))$. As a simple example, consider $iV^3_0$. We have
\be
[iV^3_0, V^2_{-1}] = 2iV^3_{-1} = iV^2_{-1}(2V^2_0 + 1) \;, \quad [V^3_0, V^2_0] = 0 \;,
\ee
so by a simple recurrence argument,
\be
({\rm ad}_{iV^3_0})^n(V^2_{-1}) = i^n V^2_{-1} (2V^2_0 + 1)^n \;.
\ee
Clearly, the adjoint action of $iV^3_0$ on $V^2_{-1}$ generates an infinite-dimensional subspace of $\hs(\lambda)$, and $\exp ({\rm ad}_{iV^3_0})(V^2_{-1})$ cannot be expressed as a finite sum of generators. A similar argument applies to the adjoint action of the generators $i^sV^s_0$ on $V^2_{-1}$. It might seem that this problem should be easily solved by considering an appropriate completion of $\hs(\lambda)$. We explain why constructing such a completion is not a simple task in Section \ref{SecCompl} below. 

Most of the time, it is hard or impossible to obtain a closed form for the power of the adjoint action of a generator. In \cite{Dixmier1973}, Dixmier provided a sufficient condition for an element to fail to exponentiate which is easy to check. The condition involves the familiar star product representation \cite{Vasiliev:1992gr}. In other terms, recall that the one-dimensional Heisenberg-Weyl algebra $A_1$ is the free algebra generated by the symbols $\{x, \partial\}$ quotiented by the relation $[\partial, x] = 1$. It coincides with the algebra of polynomial differential operators on $\R$. We have a homomorphism $\phi$ of $\mathfrak{sl}(2,\R)$ into $A_1$, given explicitly by
\be
\phi(X_+) = -\frac{1}{2}x^2 \;,\quad \phi(X_0) = x\partial + \frac{1}{2} \;, \quad \phi(X_-) = \frac{1}{2}\partial^2
\ee
inducing a homomorphism (still written $\phi$) of $\hs(\lambda)$ into $A_1$. Lemma 5.2 of \cite{Dixmier1973} can be reformulated as follows:
\begin{theorem}
\label{TheExpAdjActDix}
Let $X \in \hs(\lambda)$ and $\phi(X) = \sum_{i,j} \alpha_{ij} x^i \partial^j$. Let $r,s$ be the smallest non-negative integers such that $\alpha_{i0} = 0$ for $i > r$ and $\alpha_{0j} = 0$ for $j > s$. Assume that there are integers $\tilde{i},\tilde{j}$ such that 
\begin{enumerate}
	\item $\alpha_{\tilde{i}\tilde{j}} \neq 0$,
	\item $(\tilde{i},\tilde{j}) \neq (1,1)$,
	\item $s\tilde{i} + r\tilde{j} > rs$, or $\tilde{i},\tilde{j} \neq 0$ if $r = s = 0$.
\end{enumerate}
Then $\{({\rm ad}_{X})^n(V^2_{-1})\}_{n \in \mathbb{N}}$ generates an infinite-dimensional subspace of $\hs(\lambda)$ and $X$ does not exponentiate.
\end{theorem}
For example, we have
\be
V^3_1 = \frac{1}{8}(2x\partial^3 + 3\partial^2) \;, \quad (r,s) = (0,3) \;. 
\ee
Taking $(\tilde{i},\tilde{j}) = (1,3)$, we see that $V^3_1$ does not exponentiate. Clearly, the conditions of Theorem \ref{TheExpAdjActDix} are not necessary, as $V^3_0$ does not satisfy them.

\subsection{Completions}

\label{SecCompl}

As the discussion in the previous section has made clear, there are many elements of $\hs(\lambda)$ which do not exponentiate when we consider only finite linear combinations of generators. One might want therefore to consider a completion $\overline{\hs}(\lambda) \supset \hs(\lambda)$ allowing for certain infinite linear combinations of generators. It is reasonable to require the following properties from such a completion:
\begin{itemize}
	\item $\overline{\hs}(\lambda)$ is a Lie algebra, i.e. the Lie bracket of $\hs(\lambda)$ extended linearly to $\overline{\hs}(\lambda)$ is well-defined for any two elements of $\overline{\hs}(\lambda)$ and belongs to $\overline{\hs}(\lambda)$.
	\item The group of inner automorphisms of $\overline{\hs}(\lambda)$ contains ${\rm Aut}_{\rm in}(\hs(\lambda))$, i.e. all the elements exponentiating in $\hs(\lambda)$ exponentiate as well in $\overline{\hs}(\lambda)$. In particular, this implies that there is an action of $SL(2,\mathbbm{R})$ on $\overline{\hs}(\lambda)$ extending the corresponding adjoint action on $\hs(\lambda)$.
\end{itemize}
The most obvious way of achieving this would be to define an algebra norm on the associative algebra $B(\lambda)$ and to take the completion with respect to this norm, thereby turning it into a Banach algebra. Recall that an algebra norm is a norm such that $|XY| \leq |X||Y|$ for all elements $X,Y$, which immediately implies that the exponential of any element of finite norm has finite norm as well. Unfortunately, a simple argument \cite{Luminet1987} shows that no such norm can exist. Indeed, assuming its existence, the relation $[X_0,X_+^n] = 2n X_+^n$, we deduce that $2|X_0||X_+^n| \geq 2n|X_+^n|$ for all $n$, a contradiction.

Completions of enveloping algebras of arbitrary Kac-Moody algebras have been considered in \cite{MR735060, Alekseev:2007in}. Applied to $\mathfrak{sl}(2,\mathbbm{R})$, this amounts to consider series of the form
\be
\sum_{k,l=0}^\infty X_-^k \phi_{kl}(X_0) X_+^l \;,
\ee
where $\phi_{km}$ are analytic functions which vanish for $k-l > c$, for some positive integer $c$. The latter condition is essential to ensure that the commutator of two such series yields a series of the same type with finite coefficients. Unfortunately, it is not preserved by the adjoint action of $SL(2,\mathbbm{R})$. For instance, the rotation $\exp \theta (X_+ - X_-) \in SL(2,\mathbbm{R})$ does not have a well-defined adjoint action on this completion.

These difficulties provide us with the motivation to follow an alternative approach, which we will undertake in the next section.

\section{Exponentiation in complementary series representations}

\label{SecExpCompSeries}

Recall that we showed in Lemma \ref{LemRepASAOp} that given any irreducible representation of ${\rm sl}(2,\mathbbm{R})$ by skew-symmetric operators with $\rho(\Omega) = \frac{1}{2}(\lambda^2-1) \mathbbm{1}$, then $\rho$ provides a representation of $\hs(\lambda)$ by skew-symmetric operators. By starting from a unitary representation of $SL(2,\mathbbm{R})$, in which the generators ${\rm sl}(2,\mathbbm{R})$ are represented by skew-adjoint operators, one may hope that elements in the induced representation of $\hs(\lambda)$ are represented by skew-adjoint operators as well. The spectral theorem ensures that skew-adjoint operators admit a functional calculus and their exponentials are always well-defined unitary operators. As we will see, the story is more complicated than this, but we will find a setup in which the exponentiable elements of $\hs(\lambda)$ form a dense subspace. 

In Sections \ref{SecExpElComp} to \ref{SecPropGrHSSym}, we assume that the reader is familiar with the material reviewed in Appendix \ref{SecRevMathRes}.

\subsection{Some facts about the representation theory of $SL(2,\mathbbm{R})$}

\label{SecBasFactRepTh}

The following material can for instance be found in Chapter II of \cite{knapp2001representation}. $SL(2,\mathbbm{R})$ is the group of real $2\times 2$ matrices with unit determinant. Let $\lambda = u + iv \in \C$ and 
\be
\gamma = \left( \begin{array}{cc} a & b \\ c & d \end{array} \right) \in SL(2,\R) \;.
\ee
Given a complex valued function $f$ on $\R$, consider the right action 
\be
\label{EqActNUPRSSL2R}
\rho^{\pm,\lambda}(\gamma)f(x) = |-bx + d|^{-1-\lambda} \, {\rm sgn}(-bx + d)^{(1 \mp 1)/2} \, f \left(\frac{ax-c}{-bx+d}\right) \;.
\ee
For complex valued functions $f$ and $g$ on $\R$, define the following hermitian form: 
\be
(f,g)_p = \int_\R f(x) \bar{g}(x) (1 + |x|^2)^u dx \;.
\ee
Let $\mathcal{H}_{p,\lambda}$ be the Hilbert space of $L^2$-integrable functions with respect to $(.,.)_p$. The action $\rho^{\pm,\lambda}$ of $SL(2,\R)$ on $\mathcal{H}_{p,\lambda}$ forms the non-unitary principal series $\mathscr{P}^{\pm,\lambda}$ of representations of $SL(2,\R)$. These representations are unitary only for $u = 0$ (i.e. $\lambda$ imaginary). Alternatively, for $0 < \lambda < 1$ real, consider the hermitian form 
\be
\label{EqSesqFormComplSer}
(f,g)_c = \int_{\R^2} \frac{f(x) \bar{g}(y)}{|x-y|^{1-\lambda}} dx dy
\ee
and let $\mathcal{H}_\lambda$ be the associated Hilbert space of $L^2$-integrable functions. Then the action $\rho^{+,\lambda}$ on $\mathcal{H}_\lambda$ is unitary, and forms the complementary series $\mathscr{C}^\lambda$ of representations of $SL(2,\R)$.

We can readily compute the infinitesimal action of a generator $X \in \mathfrak{sl}(2,\R)$ associated to $\rho^{\pm,w}$ by
\be
\rho^{\pm,\lambda}(X) f(x) = \left. \frac{d}{dt} \rho^{\pm,\lambda}(\exp tX)f(x) \right |_{t = 0} \;, \quad t \in \mathbb{R} \;.
\ee 
For the elementary generators above, we obtain
\begin{align}
\label{EqActInfsl2PrinSer}
\rho^{\pm,\lambda}(X_+) f(x) = & \; (\lambda+1) x f(x) + x^2 \partial f(x) \;, \notag  \\
\rho^{\pm,\lambda}(X_0) f(x) = & \; (\lambda+1)f(x) + 2x \partial f(x) \;, \\
\rho^{\pm,\lambda}(X_-) f(x) = &  -\partial f(x) \;,  \notag 
\end{align}
where we wrote $\partial = \frac{d}{dx}$. Note that the infinitesimal action is independent of the label $\pm$, and we will simply denote it by $\rho^\lambda$. In the representation $\mathscr{P}^{\pm,\lambda}$, $x^\dagger = x$ and $\partial^\dagger = -\partial$, which makes it clear that the operators \eqref{EqActInfsl2PrinSer} are skew-symmetric provided $\lambda$ is imaginary. In the representation $\mathscr{C}^\lambda$, we have formally $\partial^\dagger = -\partial$, $x^\dagger = x + \lambda \partial^{-1}$, which allows one to check as well that \eqref{EqActInfsl2PrinSer} are skew-symmetric. 

The Casimir operator is $\rho^{\lambda}(\Omega) = \frac{1}{2}(\lambda^2-1)$. By Lemma \ref{LemRepASAOp}, we see that $\mathscr{P}^{\pm,\lambda}$ for $\lambda$ imaginary, and $\mathscr{C}^\lambda$ for $0 < \lambda < 1$, provide representations of $B(\lambda)$ and $\hs(\lambda)$ by skew-symmetric operators.

\subsection{Physical motivation}

\label{SecPhysMot}

We present here a geometrical interpretation of the higher spin symmetry which supports the physical relevance of the complementary series of representations. We gathered it in essence from \cite{Bekaert:2008sa} (see also \cite{Vasiliev:1996hn,Leigh:2014tza}). 

Consider the higher spin theories relevant to the Gaberdiel-Gopakumar duality \cite{Gaberdiel:2010pz, Gaberdiel:2012uj}. Their field content consists in an infinite tower of gauge fields of spin 2,3,..., as well as a scalar field. The collection of gauge fields is encoded in an $\hs(\lambda) \oplus \hs(\lambda)$-valued connection on the $AdS_3$ spacetime, and the mass of the scalar is related to the parameter $\lambda$ by $m^2 = \lambda^2-1$, in units in which the AdS radius is 1. In AdS, the squared mass of a scalar field can be negative without creating instabilities, as long as it is higher than the Breitenlohner-Freedman bound \cite{Breitenlohner:1982jf}, which in three dimensions reads $m^2 > -1$. As representations in the principal series have $\lambda$ imaginary, they are incompatible with the bound. However, representations in the complementary series have a chance to be relevant, as they are associated with a negative mass squared scalar compatible with the bound. In fact, in the holographic construction, we have
\be
\lambda = \frac{N}{N+k}
\ee
where $N$ and $k$ are the rank and the level of a two-dimensional conformal field theory, requiring therefore $0 \leq \lambda \leq 1$. As was mentioned in a previous footnote, there are unitary representations of $SL(2,\R)$ with $\lambda = 0$ or $1$, but from now on we will focus on the complementary series and assume that $0 < \lambda < 1$.

In order to understand better how these representations originate, we need to understand the global symmetry of the scalar field that is gauged by the $\hs(\lambda) \oplus \hs(\lambda)$-valued connection. Recall that up to global issues, $AdS_3$ can be seen as the homogeneous space 
\be
(SL(2,\mathbbm{R}) \times SL(2,\mathbbm{R}))/SL(2,\mathbbm{R}) \;,
\ee
where the action is the antidiagonal one: $g.(g_1,g_2) = (gg_1, g^{-1}g_2)$. Fixing a point $p$ on $AdS_3$ and identifying it with the coset of the identity, the Lie algebra action of ${\rm sl}(2,\mathbbm{R}) \oplus {\rm sl}(2,\mathbbm{R})$ acts on a neighborhood of $p$ by infinitesimal spacetime translations, rotations and Lorentz boosts. More precisely, the infinitesimal translations, which do not preserve $p$, are associated with axial elements of the form $(x,x) \in {\rm sl}(2,\mathbbm{R}) \oplus {\rm sl}(2,\mathbbm{R})$, while the infinitesimal rotations and boosts, which leave $p$ fixed, are associated with adjoint elements $(x,-x)$. The vector fields on $AdS_3$ associated with these transformations can be identified with the images of the left invariant vector fields of the corresponding Lie algebra elements on $SL(2,\mathbbm{R}) \times SL(2,\mathbbm{R})$. The Laplacian $\Delta$ on a homogeneous space $G/H$ is given by minus the difference of the quadratic Casimirs $C_G - C_H$, where the Casimirs are seen as differential operators of degree two on the group manifold through the identification of the Lie algebra elements with invariant vector fields. In our case, writing $C_R$, $C_L$, $C_{Ax}$ and $C_{Ad}$ for the Casimirs of the chiral left, chiral right, axial and adjoint ${\rm sl}(2,\mathbbm{R})$ subalgebras, we have $C_L + C_R = C_{Ax} + C_{Ad}$, so the Laplacian is identified with $-C_L - C_R 	+ C_{Ad} = -C_{Ax}$.

The scalar field $\phi$ satisfies the Klein-Gordon equation
\be
\label{EqEqMotScal}
(\Delta + m^2)\phi = 0 \;.
\ee
Any differential operator $D$ acting on $\phi$ and commuting with the Laplacian is an infinitesimal symmetry of the equations of motion. We are interested here only in symmetries in an infinitesimal neighborhood of a point $p$. Therefore, we see $\phi$ as defining an element in the infinite jet $\mathcal{J}^\infty_p$ at $p$, given by the collection of all its partial derivatives. (For an introduction to jets, see for instance Appendix D of \cite{Bekaert:2008sa} or \cite{2009arXiv0908.1886S}.) A linear differential equation defines a linear subspace of $\mathcal{J}^\infty_p$, encoding the linear relations between the partial derivatives of the solutions. In particular, there is a subspace $EM_{m,p} \subset \mathcal{J}^\infty_p$ corresponding to the equations of motion \eqref{EqEqMotScal}. This is in essence the geometric interpretation of the unfolded formalism of Vasiliev, used to write down the equations of motion of higher spin theories. (Compare for instance with Section 2 of \cite{Vasiliev:1996hn}.)

A differential operator corresponds to a linear map of $\mathcal{J}^\infty_p$ to itself, which is determined by the action of the differential operator on functions at $p$. We see therefore that the ``infinitesimal symmetries of the equations of motion at $p$'' can be pictured as those endomorphisms of $\mathcal{J}^\infty_p$ coming from differential operators and preserving $EM_{m,p}$.

The Killing vector fields associated to a fixed basis of ${\rm sl}(2,\mathbbm{R}) \oplus {\rm sl}(2,\mathbbm{R})$ commute with $\Delta$ and are symmetries of the equations of motion. They generate an associative subalgebra of the algebra of differential operators isomorphic to $U({\rm sl}(2,\mathbbm{R}) \oplus {\rm sl}(2,\mathbbm{R}))$. Acting on functions at $p$, we get a (reducible) representation of $U({\rm sl}(2,\mathbbm{R}) \oplus {\rm sl}(2,\mathbbm{R}))$ on $\mathcal{J}^\infty_p$. As was mentioned above, the Laplacian coincides with $-C_{Ax}$, and $EM_{m,p}$ is the subspace on which $C_{Ax} = m^2\mathbbm{1} = (\lambda^2-1)\mathbbm{1}$. As the Casimir is central, $EM_{m,p}$ provides a subrepresentation of $U({\rm sl}(2,\mathbbm{R}) \oplus {\rm sl}(2,\mathbbm{R}))$. Moreover, we are considering a scalar field, for which $C_{Ad} = 0$. We therefore learn that the Casimirs of each of the chiral ${\rm sl}(2,\mathbbm{R})$ take the value $\frac{1}{2}(\lambda^2-1)$ on $EM_{m,p}$. As $EM_{m,p}$ obviously integrates to a representation of $SL(2,\mathbbm{R}) \times SL(2,\mathbbm{R})$, it has to be a direct sum of representations in the complementary series with parameter $\lambda$. In addition, we see that given the values of the Casimirs, the representation of $U({\rm sl}(2,\mathbbm{R}) \oplus {\rm sl}(2,\mathbbm{R}))$ on $EM_{m,p}$ factors through a representation of $\hs(\lambda) \oplus \hs(\lambda)$, explaining the appearance of the higher spin algebra.

It would be worth exploring these ideas further. But for now, we take this argument as evidence that the complementary series of representation of $SL(2,\mathbbm{R})$ plays a central role in the physics of the higher spin field theory, and that it is the correct setup to look for a way of exponentiating $\hs(\lambda)$.

\subsection{Faithfulness of the complementary series representations}

\label{SecFaithCompSerRep}

The aim of this section is to show that the representation $\rho^\lambda$ of $\hs(\lambda)$ is faithful, i.e. that the kernel of $\rho^{\lambda}$ vanishes. This point is essential, as our aim is to define the higher spin symmetry group by exponentiating $\rho^{\lambda}$.

$\rho^{\lambda}$ in \eqref{EqActInfsl2PrinSer} defines a homomorphism $\phi_\lambda$ of $\mathfrak{sl}(2,\R)$ and $\hs(\lambda)$ into the Heisenberg-Weyl algebra $A_1$. This homomorphism is not equivalent to the one encountered in Section \ref{SecNonExpElAdj} \cite{2005math......4224R}. There is a natural grading on $A_1$ assigning degree 1 to $x$ and degree $-1$ to $\partial$. If we endow $\mathfrak{sl}(2,\R)$ with the grading assigning degree $1$ to $X_+$, degree $0$ to $X_0$ and degree $-1$ to $X_-$, $\phi_\lambda$ preserves the gradings. 

\begin{lemma}
$\phi_\lambda$ is an injective homomorphism of $B(\lambda)$ into $A_1$. 
\end{lemma}
\begin{proof}
As $\phi_\lambda$ preserves the grading described above, we only have to check that no nontrivial linear combination of $V^s_0$ lies in the kernel of $\phi_\lambda$. Let us write all the elements of $A_1$ as sums of monomials of the form $x^n \partial^m$, $n,m \in \N$. We can then write
\be
\phi_\lambda(X_+^{s-1}) = c x^{2s-2} \partial^{s-1} + ...
\ee
for $c \in \R$, $c \neq 0$ and the dots denote a sum of monomials involving powers of $\partial$ smaller than ${s-1}$. We get from \eqref{EqDefVSBasisHSL}
\be
\phi_\lambda(V^s_0) = c' x^{s-1} \partial^{s-1} + ...
\ee
and $\phi_\lambda(V^s_0)$ is linearly independent from the set $\{\phi_\lambda(V^{s'}_0)\}_{s' < s}$.
\end{proof}
As the representations of $A_1$ on $\mathcal{H}_\lambda$ are faithful, we deduce that the complementary series representations are faithful representations of $\hs(\lambda)$.

For future convenience, we list the images of the generators of $\hs(\lambda)$ in the Heisenberg-Weyl algebra, obtained from \eqref{EqActInfsl2PrinSer}:
$$
\rho^\lambda(V^2_1) = - \partial \;, \quad
\rho^\lambda(V^2_0) = \frac{\lambda+1}{2} + x \partial \;, \quad
\rho^\lambda(V^2_{-1}) = - (\lambda+1) x - x^2 \partial \;,
$$ 
$$
\rho^\lambda(iV^3_2) = i\partial^2 \;,  \quad
\rho^\lambda(iV^3_1) = -i\frac{\lambda+2}{2} \partial - i x \partial^2 \;,
$$
\be
\label{EqActInfhslambda}
\rho^\lambda(iV^3_0) =  i\frac{(\lambda+1)(\lambda+2)}{6} + i(\lambda+2) x \partial + i x^2 \partial^2 \;,
\ee
$$
\rho^\lambda(iV^3_{-1}) = -i \frac{(\lambda+1)(\lambda+2)}{2}x - 3i \frac{(\lambda + 2)}{2} x^2 \partial - i x^3 \partial^2 \;,
$$
$$
\rho^\lambda(iV^3_{-2}) = i(\lambda+1)(\lambda+2)x^2 + 2i(\lambda+2)x^3 \partial + ix^4 \partial^2  \;. \notag 
$$

\subsection{The circle model for the complementary series representations}

\label{SecCircModCompSerRep}

We will refer to the model for representations in the complementary series in terms of functions on the real axis as the \textit{line model}. There is another model for representations in the complementary series, in terms of the space of functions on the circle (see for instance the original work \cite{Bargmann1947}, Section 8), which we will refer to as the \textit{circle model}. It will prove useful when performing explicit exponentiations in Section \ref{SecExFinHSTrans}.

To derive the circle model, we perform a stereographic mapping of the real axis onto the circle of unit modulus complex numbers,
\be
x = i\frac{z - 1}{z + 1} \;, \quad z = -\frac{x + i}{x - i}
\ee
and parameterize the unit circle with a periodic variable $z = \exp i \theta$. Altogether, this amounts to a change of variable
\be
x = - \tan (\theta/2) \;.
\ee
We have to compute the image of the hermitian form \eqref{EqSesqFormComplSer} under this change of variable. Write $y = - \tan (\eta/2)$. Using the standard doubling and sum/difference formulas for trigonometric functions, we have
\be
|x-y| = |\tan(\eta/2) - \tan(\theta/2)| = \frac{|1-\cos(\theta - \eta)|^{1/2}}{\sqrt{2}\cos(\theta/2)\cos(\eta/2)} \;,
\ee
so the measure factor becomes
\be
|x-y|^{\lambda-1} dx dy = \frac{1}{\sqrt{2}}|1-\cos(\theta - \eta)|^{\frac{\lambda-1}{2}} \cos(\theta/2)^{-1-\lambda} \cos(\eta/2)^{-1-\lambda} d\theta d\eta \;.
\ee
We map a function $f(x)$ on the real axis to a function $F(\theta)$ on $S^1$ by
\be
\label{EqMapFuncModCompRep}
F(\theta) = (\cos(\theta/2))^{-1-\lambda} f(x(\theta)) \;,
\ee
so that
\be
\label{EqSesqFormCirclMod}
(f,g)_c = (F,G)_c = \frac{1}{\sqrt{2}} \int_{S^1 \times S^1}F(\theta) \bar{G}(\eta) |1-\cos(\theta - \eta)|^{\frac{\lambda-1}{2}} d\theta d\eta \;.
\ee
We recover the Hilbert space $\mathcal{H}_\lambda$ as the space of square integrable functions on the circle with respect to the hermitian form above. The $SL(2,\mathbbm{R})$ action turns into the natural action via Moebius transformations of $SU(1,1)$ on the unit circle in the complex plane.

$\partial = \frac{d}{dx}$ is mapped under \eqref{EqMapFuncModCompRep} to 
\be
-(1 + \cos \theta) \partial_\theta + s \sin \theta \;,
\ee
where $\partial_\theta = \frac{d}{d\theta}$ and $s = (\lambda+1)/2$. The ${\rm sl}(2,\mathbbm{R})$ generators read
\be
\sigma^\lambda(X_+) = -s\sin \theta - (1 - \cos \theta) \partial_\theta \;,
\ee
\be
\sigma^\lambda(X_0) = 2s \cos \theta + 2 \sin \theta \partial_\theta \;,
\ee
\be
\sigma^\lambda(X_-) =  -s\sin \theta + (1 + \cos \theta) \partial_\theta \;.
\ee
Remark the simple form of the compact generator of $SL(2,\mathbbm{R})$, corresponding to infinitesimal rotations of the unit circle:
\be
\sigma^\lambda\left(\frac{1}{2}(X_+ - X_-)\right) = -\partial_\theta \;.
\ee
It is easy to work out the expression of any generator of $\hs(\lambda)$ in the circle model. Of course, just like in the line model, all the elements of $\hs(\lambda)$ are represented by differential operators that are skew-symmetric with respect to the hermitian form \eqref{EqSesqFormCirclMod}.

\subsection{Exponentiable elements}

\label{SecExpElComp}

An essentially skew-adjoint operator admits a unique skew-adjoint extension, and the latter generates a one-parameter subgroup of $U(\mathcal{H}_\lambda)$, the group of unitary transformations of $\mathcal{H}_\lambda$ in the strong operator topology (see Appendix \ref{SecFuncCalcExp}). Therefore, any element of $\hs(\lambda)$ represented by an essentially skew-adjoint operator on $\mathcal{H}_\lambda$ exponentiates. Certain non-essentially skew-adjoint operators can be exponentiated as well, but a choice of extension has to be made, and we do not know a way of picking a particular extension. We discuss this point further in Appendix \ref{SecExtNonEssSAOp}, but in the rest of the paper, we will consider an element of $\hs(\lambda)$ to be exponentiable if and only if it is represented by an essentially skew-adjoint operator.

In Appendix \ref{SecCritSAEnvAlg}, we review sufficient criterions for an element of $\hs(\lambda)$ to be represented on $\mathcal{H}_\lambda$ by an essentially skew-adjoint operator and to exponentiate. They can be summarized as follows.
\begin{enumerate}
	\item All the elements of ${\rm sl}(2,\mathbbm{R}) \subset \hs(\lambda)$ exponentiate. 
	\item Any element of $\hs(\lambda)$ expressible as a complex polynomial of a single generator of ${\rm sl}(2,\mathbbm{R})$ exponentiates.
	\item Any elliptic element of $\hs(\lambda)$ exponentiates. (An elliptic element of an enveloping algebra is an element such that the corresponding differential operator on the Lie group is elliptic, see Appendix \ref{SecCritSAEnvAlg}.)
	\item Any element commuting with an elliptic element of $\hs(\lambda)$ exponentiates.
\end{enumerate}
The first criterion is obvious, as we started from a representation of $SL(2,\mathbbm{R})$. From the second criterion, we learn in particular that $i^sV^{s}_{s-1}$ and $i^sV^{s}_{-s+1}$ exponentiate. This implies that all the elements of $\hs(\lambda)$ which did exponentiate in the adjoint representation exponentiate in the complementary series representation (see Theorem \ref{ThGrInnAutHSL}). 

The third criterion tells us that the subset of exponentiable elements is dense in $\hs(\lambda)$. Indeed, given $X \in \hs(\lambda)$ a finite sum of generators with maximal spin $s$, it is always possible to find an elliptic element $Y$ in $\hs^{(s')}(\lambda)$ for some $s' > s$. Then $X + \epsilon Y$ is elliptic for all $\epsilon > 0$. So while not all of the elements of $\hs(\lambda)$ exponentiate, as will be shown in Section \ref{SecV31IsNotEssSA}, there are always exponentiable elements in an arbitrarily small neighborhood of any element.

Interestingly, the criterion of being elliptic is completely independent of the choice of unitary representation of $SL(2,\mathbbm{R})$. This suggests that there might exist a definition of the higher spin symmetry group which does not require the use of any representation, unlike the one we will give in Section \ref{SecPropGrHSSym}.

Even more interestingly, from the picture we developed in Section \ref{SecPhysMot}, an element of $\hs(\lambda)$ is elliptic if and only if the corresponding differential operator on spacetime, representing an infinitesimal symmetry of the Klein-Gordon equation at $p \in AdS_3$, is elliptic. This suggests that the reason why certain elements of $\hs(\lambda)$ exponentiate and others do not can be understood from a spacetime point of view.

\subsection{A non-exponentiable element}

\label{SecV31IsNotEssSA}

We prove here that $iV^3_1$ does not exponentiate in the complementary series representation. To this end, we show that while $i\rho^\lambda(V^3_1)$ is a skew-symmetric differential operator with respect to \eqref{EqSesqFormComplSer}, it is not skew-adjoint. This section is a bit technical and skipping it will not impair the understanding of the rest of the paper.

We have
\be
T = \rho^\lambda(iV^3_1) = -ix \partial^2 - i\frac{\lambda+2}{2} \partial
\ee
and the domain of $T$ is the space of smooth functions on $\mathbbm{R}$ whose norm \eqref{EqSesqFormComplSer} is finite. Recall the criterion for skew-adjointness presented in Appendix \ref{SecCritSAEnvAlg}: a skew-symmetric operator $T$ is essentially skew-adjoint if and only if there is no $f \in \mathcal{D}(T^\ast)$ such that $T^\ast f = f$ or $T^\ast f = -f$. The possible lack of essential skew-adjointness of $T$ is therefore equivalent to the existence of weak solutions of the differential equation
\be
\label{EqDiffEqV31}
(\epsilon - T)f = ix \partial^2 f + i\frac{\lambda+2}{2} \partial f + \epsilon f = 0 \;, \quad \epsilon = \pm 1
\ee
in $\mathcal{H}_\lambda$, satisfying
\be
\label{EqDiffEqV31Weak}
\left(f, (\epsilon + T) g \right)_c = 0 \;,
\ee
for all smooth test functions $g$. 

We start by studying formal solutions of \eqref{EqDiffEqV31}. Substituting $f(x) = |x|^{-\lambda/4} g(x)$ and performing the change of variable $y = (1-\epsilon i)\sqrt{2|x|}$ turns \eqref{EqDiffEqV31} into the Bessel differential equation for the variable $y$. Writing ${\rm sgn}(x)$ for the sign of $x$ and $J_\alpha$ for the Bessel function of the first kind, the two linearly independent solutions on the intervals $(-\infty, 0)$ and $(0,\infty)$ read
\be
\label{EqElSolDiffEqV31}
f_1(x) = |x|^{-\lambda/4}J_{-\lambda/2}\big((1+i)\sqrt{2|x|}\big) \;, \quad f_2(x) = |x|^{-\lambda/4}J_{\lambda/2}\big((1+i)\sqrt{2|x|}\big)
\ee
for $\epsilon  = - {\rm sgn}(x)$ and 
\be
\label{EqElSolDiffEqV31b}
f_1(x) = |x|^{-\lambda/4}J_{-\lambda/2}\big((1-i)\sqrt{2|x|}\big) \;, \quad f_2(x) = |x|^{-\lambda/4}J_{\lambda/2}\big((1-i)\sqrt{2|x|}\big)
\ee
for $\epsilon  = {\rm sgn}(x)$. 

We now have to figure out if there are linear combinations of these formal solutions that belong to $\mathcal{H}_\lambda$ and satisfy equation \eqref{EqDiffEqV31Weak}. To this end, it is useful to define the Bessel functions of the second kind and the Hankel functions:
\be
Y_\nu(z) = \frac{J_\nu(z)\cos \pi \nu -  J_{-\nu}(z)}{\sin \pi \nu} \;, \quad H^{+}_\nu(z) = J_\nu(z) + i Y_\nu(z) \;, \quad H^{-}_\nu(z) = J_\nu(z) - i Y_\nu(z) \;
\ee
for $\nu > 0$. The Hankel functions follow the asymptotics 
\be
H_\nu^{+}(z) \sim \sqrt{\frac{2}{\pi z}}\exp\left(i\left(z-\frac{\nu\pi}{2}-\frac{\pi}{4}\right)\right) \mbox{ for } |\arg z| < \pi
\ee
\be
H_\nu^{-}(z) \sim \sqrt{\frac{2}{\pi z}}\exp\left(-i\left(z-\frac{\nu\pi}{2}-\frac{\pi}{4}\right)\right) \mbox{ for } |\arg z| < \pi
\ee
for large $|z|$ and $H_\nu^{\pm}(z) \sim z^{-\nu}$ for small $|z|$. We see that for $\epsilon = 1$, $H^-_{\lambda/2}((1-i)\sqrt{2|x|})$ is a square summable solution on $(0,\infty)$ and $H^+_{\lambda/2}((1+i)\sqrt{2|x|})$ is a square summable solution on $(-\infty,0)$. We can extend them to all of $\mathbbm{R}$ by zero and write $f^>_1$ and $f^<_1$ for the resulting functions. Similarly, for $\epsilon = -1$, we have a square summable solution $H^+_{\lambda/2}((1+i)\sqrt{2|x|})$ on $(0,\infty)$ and $H^-_{\lambda/2}((1-i)\sqrt{2|x|})$ on $(-\infty,0)$. We extend them as well to $\mathbbm{R}$ by zero and write $f^>_{-1}$ and $f^<_{-1}$ for the resulting functions. $f^>_\epsilon$ and $f^<_\epsilon$ are linearly independent solutions in $\mathcal{H}_\lambda$ of \eqref{EqDiffEqV31} outside $x = 0$. We need now to understand which of their linear combinations satisfy \eqref{EqDiffEqV31Weak}. 

Using the explicit expression \eqref{EqSesqFormComplSer} for the norm on $\mathcal{H}_\lambda$, we find that \eqref{EqDiffEqV31Weak} can be written
\be
\label{EqRewEqWeakSolV31}
\int_{-\infty}^{\infty} dx f(x) \left(\epsilon + i \partial^2 x - i \frac{\lambda+2}{2}\partial \right) \tilde{g}(x) \;,
\ee
where $\tilde{g}(x)$ is the smooth function defined by
\be
\tilde{g}(x) = \int_{-\infty}^{\infty} dy \bar{g}(y) |x-y|^{\lambda-1} \;.
\ee
We fix $\epsilon$ and take $f$ to be a linear combination of $f^<_\epsilon$ and $f^>_\epsilon$. We take a small $a > 0$ and we decompose the integral in \eqref{EqRewEqWeakSolV31} into integrals over $(-\infty, -a]$, $(-a,a)$ and $[a, \infty)$. As $f$ is square summable and $\tilde{g}$ is smooth, the integral over $(-a,a)$ goes to zero as $a \rightarrow 0$. Integrating by part on the two remaining domains, we find integrands proportional to $(\epsilon - T)f = 0$, so only the boundary terms at $-a$ and $a$ might prevent \eqref{EqRewEqWeakSolV31} to vanish. They can be computed explicitly and read respectively
\be
-i\lambda a^{-\lambda/2} \tilde{g}(-a) - i a^{-\lambda/2 + 1} \partial \tilde{g}(-a) \mbox{ for } f = f^<_\epsilon
\ee
and
\be
i\lambda a^{-\lambda/2} \tilde{g}(a) - i a^{-\lambda/2 + 1} \partial \tilde{g}(a) \mbox{ for } f = f^>_\epsilon \;.
\ee
The second terms clearly tend to zero as $a \rightarrow 0$, which the first terms will cancel only for $f = f^>_\epsilon + f^<_\epsilon$ up to scalar multiples. We therefore find one weak solution for each $\epsilon = \pm 1$. The deficiency indices of $\rho(iV^3_1)$ are $(1,1)$, $\rho(iV^3_1)$ is not a skew-adjoint operator and $V^3_1$ does not exponentiate in the complementary series representation.

\subsection{The higher spin symmetry group}

\label{SecPropGrHSSym}

We can now define the higher spin symmetry group $HS(\lambda)$:\vspace{.2cm}

\noindent\textbf{Definition}: $HS(\lambda)$ is the subgroup of $U(\mathcal{H}_\lambda)$ spanned by the one-parameter subgroups generated by the exponentiable elements of $\hs(\lambda)$.\vspace{.2cm}

Our discussion in Section \ref{SecPhysMot} lets us hope for a more conceptual definition. Recall that we characterized $\hs(\lambda)$ as the Lie algebra of differential operators at a point $p$ in $AdS_2$ preserving the subspace $E_{m,p}$ of the infinite jet $\mathcal{J}^\infty_p$ corresponding to the Klein-Gordon differential equation at $p$. Such differential operators describe infinitesimal symmetries of the Klein-Gordon equation. This suggests that $HS(\lambda)$ can be defined as a group of transformations of $\mathcal{J}^\infty_p$ preserving $E_{m,p}$. If this definition can be made precise, it will be a useful handle on the higher spin symmetry group. We hope to come back to this question in the future, and will adhere to the practical definition above for the rest of the paper.

We now make a few remarks about $HS(\lambda)$. Again, each of these questions would deserve detailed studies, but we leave this for future work. 

\paragraph{Topology} Recall that a \textit{topological group} is a group endowed with a topology compatible with the group structure, i.e. a topology in which the multiplication and the inverse map are continuous. The group $U(\mathcal{H}_\lambda)$ of unitary operators on $\mathcal{H}_\lambda$ can be given several sensible topologies, but we are using here the \textit{strong operator topology} (see Section VI.1 of \cite{reed1980methods}). 
This topology is natural because it makes the representation map $SL(2,\mathbbm{R}) \rightarrow U(\mathcal{H}_\lambda)$ continuous. (This is not the case with the norm topology on $U(\mathcal{H}_\lambda)$, for instance.) The strong operator topology on $U(\mathcal{H}_\lambda)$ endows $HS(\lambda)$ with the structure of a topological group.

\paragraph{Lack of Lie group structure} A \textit{Lie group} is a topological group endowed with a smooth manifold structure, and for which the multiplication and inverse are smooth maps. Unfortunately, $HS(\lambda)$ is not a Lie group in the strong operator topology. This can be seen by adapting an argument due to Neeb and appearing in \cite{2013arXiv1309.5891S}. Assume that there is an countable orthonormal basis of $\mathcal{H}_\lambda$, such that the commutative group $K = U(1)^\mathbbm{N}$ of unitary operators diagonal in this basis is a subgroup of $HS(\lambda)$. We will see in Section \ref{SecExFinHSTrans} that $HS(\lambda)$ admits such a subgroup. The induced topology on $K$ is the product topology. Open sets in this topology are of the form $\pi_{i \in \mathbbm{N}} U_i$, where $U_i$ are open sets of $U(1)$ which are different from $U(1)$ only for finitely many $i$'s. A basic property of the product topology is that products of compact sets are compact, which implies that $K$ is compact in this topology, hence it is also locally compact. If $K$ was a Lie group, it would be an infinite-dimensional manifold and have a local model in terms of an infinite-dimensional Hausdorff topological vector space, which can never be locally compact. So no such local model exists for $K$, and neither do they for $HS(\lambda)$.

As far as we are aware, there might be a better behaved topology on $HS(\lambda)$ that would turn it into a Lie group. While no counterexample is available, it is not known in the infinite-dimensional case if a Lie group necessarily comes with an exponential map (see Section 2 of \cite{2006math......2078G}).

We need $HS(\lambda)$ to be a Lie group if we want picture the higher spin gauge field in the standard way, as a connection on a principal bundle. While there is no problem defining principal $G$-bundles for any group $G$, we need the principal bundle to be a smooth manifold in order to speak about connections/gauge fields. This requires $G$ to be itself a smooth manifold, i.e. a Lie group.

\paragraph{Adjoint action} A related problem concerns the adjoint action of $HS(\lambda)$ on $\hs(\lambda)$. Let $g \in HS(\lambda)$. $g$ is a unitary operator on $\mathcal{H}_\lambda$. Then, for $y \in \hs(\lambda)$, $g \rho^\lambda(y) g^{-1}$ is a well-defined unbounded operator on $\mathcal{H}_\lambda$. However, it will in general not be a finite order differential operator and will be expressible in terms of the operators $\rho^\lambda(V^s_n)$ only as a formal series. This suggests that the image of the adjoint action of $HS(\lambda)$ on $\rho^\lambda(\hs(\lambda))$ can be seen as a completion $\overline{\hs}(\lambda)$ of $\hs(\lambda)$. We do not know how to characterize $\overline{\hs}(\lambda)$ independently. Note also that as the elements of $\overline{\hs}(\lambda)$ are unbounded operators, it is not guaranteed that their domains overlap, so a priori nothing ensures that $\overline{\hs}(\lambda)$ carries a Lie bracket. The theorems available only show that there is a common dense domain of definition for elements in $\rho^\lambda(\hs(\lambda))$ and that it is preserved by the action of $SL(2,\mathbbm{R})$ (see Appendix \ref{SecApRepEnvAlgUnbOp} and Section 10.1 of \cite{schmudgen1990unbounded}). What needs to be shown is that the common domain of definition is preserved as well by the action of $HS(\lambda)$, which is not obvious.

\subsection{The Euclidean higher spin symmetry group}

\label{SecEuclSymGrp}

Solutions of the higher spin theory are sometimes more conveniently constructed in the Euclidean version of the theory. This is typically the case for black holes. It is therefore also interesting to define the higher spin symmetry group of the Euclidean higher spin theory. Another reason to be interested in the Euclidean group is that it is also the gauge group for the higher spin theory formulated on three-dimensional de Sitter space. 

Recall the situation in pure gravity. On Lorentzian $AdS_3$, pure gravity can be formulated as a Chern-Simons theory with gauge group $SO^+(2,2) \simeq (SL(2,\mathbbm{R}) \times SL(2,\mathbbm{R}))/\mathbb{Z}_2$, while on Euclidean $AdS_3$, the gauge group is $SO^+(3,1) \simeq PSL(2,\mathbbm{C})$. As a real Lie algebra, $\mathfrak{sl}(2,\C)$ has generators $X_+$, $X_-$, $X_0$ and
\be
Y_+ = \left( \begin{array}{cc} 0 & i \\ 0 & 0 \end{array} \right) \;, \quad Y_- = \left( \begin{array}{cc} 0 & 0 \\ i & 0 \end{array} \right) \;, \quad Y_0 = \left( \begin{array}{cc} i & 0 \\ 0 & -i \end{array} \right) \;.
\ee
It has two quadratic Casimirs, given by
\be
\Omega^E_1 = \frac{1}{4}\left(\frac{1}{2}X_0^2 + X_+ X_- + X_- X_+ - \frac{1}{2} Y_0^2 - Y_+ Y_- - Y_- Y_+\right) \;,
\ee
\be
\Omega^E_2 = \frac{1}{4} \left( X_0 Y_0 + X_+ Y_- + Y_- X_+ + Y_+ X_- + X_- Y_+ \right)
\ee
We emphasize that we always consider $\mathfrak{sl}(2,\C)$ as a real Lie algebra, so we can in particular consider its complexification $(\mathfrak{sl}(2,\C))|_\mathbb{C}$. In $(\mathfrak{sl}(2,\C))|_\mathbb{C}$ we have two commuting $\mathfrak{sl}(2,\R)$ subalgebras generated by the elements
\be
\label{EqGenCompSl2R}
\frac{1}{2}(X_j \pm i Y_j)
\ee
They are the Wick rotated images of the two $\mathfrak{sl}(2,\R)$ algebras generating the Lorentzian $AdS_3$ Poincaré group. It is easy to check that the quadratic Casimirs of these $\mathfrak{sl}(2,\R)$ algebras are given by $\Omega^E_1 \pm i \Omega^E_2$. Consequently, we expect that in the Euclidean continuation of the $\hs(\lambda)$ theory,
\be
\label{EqRelCasEuclHSAlg}
\Omega^E_1 = \frac{1}{2}(\lambda^2-1) \;, \quad \Omega^E_2 = 0 \;.
\ee
The Euclidean higher spin algebra $\hs_E(\lambda)$ is therefore constructed as follows. We start from the real universal enveloping algebra $U(\mathfrak{sl}(2,\C))$, where it is understood that $\mathfrak{sl}(2,\C)$ is seen as a real Lie algebra. We quotient it by the ideal associated to the relations \eqref{EqRelCasEuclHSAlg}, see it as a Lie algebra by considering the commutator as the bracket, and factor out the direct summand associated to the identity element as in \eqref{EqFactDirSumId}. In order to study exponentiation in $\hs_E(\lambda)$, we need to find if $PSL(2,\mathbbm{C})$ admits unitary representations in which the Casimirs take the values \eqref{EqRelCasEuclHSAlg}.

This turns out to be the case. The complementary series representations of $SL(2,\mathbb{C})$ are unitary representations, defined as follows \cite{knapp2001representation}. The underlying Hilbert space $\mathcal{H}^E_\mu$ is the space of complex valued functions over $\mathbb{C}$ that are $L^2$-summable with respect to the scalar product 
\be
\label{EqSesqFormComplSerCpx}
(f,g)_c^{E} = \int_{\C^2} \frac{f(z_1) \bar{g}(z_2)}{|z_1-z_2|^{3-\mu}} dz_1 dz_2 \;,
\ee
for $1 < \mu < 3$. The representation map is given by
\be
\rho^{\mu}_E(\gamma)f(z) = |-bz + d|^{-1-\mu} \, f \left(\frac{az-c}{-bz+d}\right) \;. 
\ee 
where we used the same conventions as in \eqref{EqActNUPRSSL2R}, with $\gamma \in SL(2,\C)$. It is easy to check that the central element $-\bm{1}$ is represented trivially, hence the representation above is also a representation of $PSL(2,\C)$. This representation is \emph{not} holomorphic, as a holomorphic representation could not be unitary. This fact was at the origin of a confusion in a previous version of this paper, where we claimed that the techniques above would not apply to $\hs_E(\lambda)$. In fact, as the complementary series representations are unitary, elements of $\hs_E(\lambda)$ are represented by skew-symmetric operators, and our techniques to study exponentiation carry over to this case. 

Restricting to smooth functions, we can compute explicitly the images of the real generators of $\mathfrak{sl}(2,\C)$ through the infinitesimal representation map:
\begin{align}
\label{EqInfRepEuclHSAlg}
\rho^\mu_E(X_-) = - \partial_x \;, \; & \; \rho^\mu_E(Y_-) = - \partial_y \notag \\
\rho^\mu_E(X_0) = (\mu + 1) + 2(x\partial_x + y \partial_y) \;, \; & \; \rho^\mu_E(Y_0) = 2(x\partial_y - y \partial_x) \\
\rho^\mu_E(X_+) = (\mu+1)x + (x^2-y^2)\partial_x + 2xy\partial_y \;, \; & \; \rho^\mu_E(Y_+) = -(\mu + 1)y + (x^2-y^2)\partial_y - 2xy\partial_x \notag
\end{align}
where we wrote $z = x + iy$, $x,y \in \R$. Remark that the fact that the representation is not holomorphic translates into the fact that the infinitesimal representation map is not complex linear: for $X \in \mathfrak{sl}(2,\C)$,
\be
\rho^{\mu}_E(iX) \neq i \rho^{\mu}_E(X) \,.
\ee
From the expression above, it is easy to compute the partial differential operator associated to any element of $\hs_E(\lambda)$. One can check that the Casimirs are given by
\be
\Omega^E_1 = \frac{1}{8}(\mu^2 - 2\mu -3) \;, \quad \Omega^E_2 = 0 \;.
\ee
This implies that $\mu$ and $\lambda$ are related by $\mu = 2\lambda + 1$, which can also be checked by using \ref{EqInfRepEuclHSAlg} to represent the complex $\mathfrak{sl}(2,\mathbb{R})$ generators \eqref{EqGenCompSl2R}.

We see that the unitary range for the Lorentzian higher spin algebra, $0 < \lambda < 1$, coincides exactly with the unitary range for the Euclidean higher spin algebra $1 < \mu < 3$. We can therefore apply the same criterion for the exponentiability of elements of $\hs_E(\lambda)$: an element exponentiates if and only if it is represented by an essentially skew-adjoint operator on the Hilbert space $\mathcal{H}^E_\mu$. The sufficient criteria of Section \ref{SecExpElComp} remain valid. Remark however that we are dealing now with partial differential operators instead of ordinary differential operators. This will also be the case for higher spin algebras associated with theories in dimension higher than 3.

\section{Examples of finite higher spin gauge transformations}

\label{SecExFinHSTrans}

We investigate here certain maximal commutative subalgebras of $\hs(\lambda)$ whose generators can be diagonalized and exponentiated explicitly. These commutative subalgebras are polynomial algebras in a generator of $SL(2,\R)$, and we consider in turn generators associated to hyperbolic, parabolic and elliptic elements of $SL(2,\R)$. We deduce some information about the global topology of $HS(\lambda)$ and in the last case, we find a set of elements of the higher spin Lie algebra exponentiating to the identity of $HS(\lambda)$.

We note that global properties of the higher spin symmetry group were previously studied in \cite{Kraus:2012uf}, where a $\mathbb{Z}_4$ subgroup of the center was described. From our point of view based on exponentiation, we are only studying the adjoint form of $HS(\lambda)$, which has a trivial center. We will see in Section \eqref{SecExpCompGen} that the adjoint form has a non-trivial first homotopy group, and therefore admits covers with non-trivial centers. We will not try to determine which of these covers is the correct gauge group, although analogy with ordinary gravity would suggest a double cover.

\subsection{A subalgebra generated by a hyperbolic element}

The higher spin algebra admits a commutative subalgebra $\mathfrak{h}$ generated by $\{i^sV^s_0\}$. We showed in Section \ref{SecNonExpElAdj} that none of these generators exponentiate in the adjoint representation, apart from $V^2_0$. But any element of $\mathfrak{h}$ can be expressed as a polynomial in the hyperbolic generator $V^2_0 = \frac{1}{2}X_0$. (See for instance equation (3.4) in \cite{Pope:1989sr}.) Criterion 2 in Section \ref{SecExpElComp} then implies that every element in $\mathfrak{h}$ exponentiates in the complementary series representation.

In order to get some insight about the global structure of the subgroup of $HS(\lambda)$ generated by $\mathfrak{h}$, remark that the spectrum of $\rho^\lambda(V^2_0) =  \frac{\lambda+1}{2} + x \partial$ covers the whole imaginary axis. Indeed, the function
\be
f^0_\alpha(x) = x^{\left(i \alpha - \frac{\lambda+1}{2} \right)} \;, \quad \alpha \in \mathbbm{R}
\ee
is an eigenfunction with eigenvalue $i\alpha$. This shows that the subgroup $H \subset HS(\lambda)$ generated by $\mathfrak{h}$ is a countable infinite product $\mathbbm{R}^\mathbbm{N}$. (Compare with Section \ref{SecExpCompGen} below.) 

\subsection{Strictly nilpotent subalgebras}

We have a subalgebra $\mathfrak{n}_+$ of strictly nilpotent elements, given by linear combinations of the generators $\{i^sV^s_{s-1}\}$. $\mathfrak{n}_+$ is the subalgebra of skew-hermitian elements in the polynomial algebra of the parabolic generator $V^2_1 = X_-$. Again, Criterion 2 in Section \ref{SecExpElComp} implies that any element in $\mathfrak{n}_+$ exponentiates in the complementary series representation.

In the line model, the generators are represented by $\rho^\lambda(i^s V^s_{s-1}) = i^s \partial^s$. $\rho^\lambda(i^s V^s_{s-1})$ can easily be diagonalized by the functions
\be
f^+_\alpha(x) = \exp i \alpha x \;, \quad \alpha \in \mathbbm{R}.
\ee
The subgroup $N_+ = \exp \mathfrak{n}_+ \subset HS(\lambda)$ is again isomorphic to $\mathbbm{R}^\mathbbm{N}$. In the circle model, the eigenfunctions read
\be
F^+_\alpha(\theta) = (\cos(\theta/2))^{-1-\lambda} \exp \left(-i\alpha \tan (\theta/2) \right)
\ee
The similarity between the expressions for $\sigma^\lambda(X^+)$ and $\sigma^\lambda(X^-)$ allows us to guess the eigenfunctions in the circle model of the commutative subalgebra $\mathfrak{n}_-$ generated by $\{i^sV^s_{-s+1}\}$. They read
\be
F^-_\alpha(\theta) = (\sin(\theta/2))^{-1-\lambda} \exp \left(i\alpha \cot (\theta/2) \right) \;.
\ee
One can perform the change of variable described in Section \ref{SecCircModCompSerRep} to find the corresponding expression for the eigenfunctions in the line model.

\subsection{Compact subalgebras and the BTZ holonomies}

\label{SecExpCompGen}

We now come to an interesting commutative subalgebra of $\hs(\lambda)$, namely the subalgebra $\mathfrak{r}$ generated by the polynomials in the elliptic generator $X_R := \frac{1}{2}(X_+ - X_-) = -\frac{1}{2}(V^2_1 + V^2_{-1})$. Again, Criterion 2 of Section \eqref{SecExpElComp} implies that any element in $\mathfrak{r}$ exponentiates in the complementary series representation.

This subalgebra is best studied in the circle model of the complementary series of representations, as we have $\sigma^\lambda(X_R) = -\partial_\theta$. Elements of $\mathfrak{r}$ are straightforwardly diagonalized by the functions
\be
F^R_p = \exp ip \theta \;, \quad p \in \mathbbm{Z} \;.
\ee
We see here an interesting phenomenon. As the spectrum of $\sigma^\lambda(X_R)$ is discrete and integral, given any polynomial $P$ with integer coefficients, $2\pi P(X_R)$ exponentiates to the identity $\mathbbm{1} \in HS(\lambda)$. The subalgebra $\mathfrak{r}$ therefore generates a commutative subgroup $R \subset HS(\lambda)$ isomorphic to a countable direct product $U(1)^\mathbbm{N}$.

Elements of $\hs(\lambda)$ exponentiating to $\mathbbm{1}$ are important ingredients in the construction of spherically symmetric solutions of the associated higher spin theory, and especially black holes. Indeed, it is natural to take the base 3-manifold to be an infinite solid cylinder (in the case of AdS), or a solid torus (in the case of thermal AdS or black hole solutions). The radial dependence of the higher spin connection can be fixed by a choice of gauge \cite{oai:arXiv.org:1008.4744,oai:arXiv.org:1107.0290} and one may look for connections which are constant in the remaining directions. Of course, the holonomy of such a connection along the contractible circular direction has to be trivial. This implies that the connection integrated along this direction has to exponentiate to $\mathbbm{1}$.
\footnote{Note that it is often claimed in the literature that the connection only has to exponentiate to an element of the center. This is at first sight slightly confusing, because a connection cannot be smooth if its holonomy along a contractible loop is different from the identity. The confusion is solved by a careful consideration of the gauge groups involved. For instance, in the case of Lorentzian AdS solutions of ordinary 3d gravity, the gauge group is $SO(2,2) = SL(2,\mathbb{R}) \times SL(2,\mathbb{R})/\mathbb{Z}_2$. The components of the connection along both $SL(2,\mathbb{R})$ subfactors integrate to $-\mathbbm{1}$, the non-trivial element of the center of $SL(2,\mathbb{R})$. However, the full connection integrates to $\mathbbm{1}$, because of the quotient by $\mathbb{Z}_2$ in the definition of the gauge group. Similarly, in Euclidean signature, gauge group is $SO(3,1) = SL(2,\mathbb{C})/\mathbb{Z}_2$. The BTZ black hole connection exponentiates to $-\mathbbm{1}$ in $SL(2,\mathbb{C})$, which is mapped to $\mathbbm{1}$ in the quotient. The additional sign in the holonomy when considering the spin cover of the gauge group comes from the fact that the bounding spin structure on $S^1$ is the non-trivial double cover. As the center of $SL(2,\mathbb{R})$ is represented trivially in the complementary representation, the present analysis amounts to taking $HS(\lambda)$ and $HS_E(\lambda)$ to the be the adjoint forms, with trivial center, which is why we focus on Lie algebra elements exponentiating to the identity.}


In many cases, including black holes, these solutions are constructed in the Euclidean version of the theory. This means that the relevant gauge group is $HS_E(\lambda)$ (see Section \ref{SecEuclSymGrp}). The above analysis of the commutative subgroups can be carried out for $HS_E(\lambda)$ as well, and one finds that polynomials in a generator $X$ of $\mathfrak{sl}(2,\C)$ generate a $\R^\N$ or $U(1)^\N$ subgroup depending on whether $X$ is compact or non-compact.

In particular, $2\pi P(X_R)$ exponentiates to $\mathbbm{1}$ in $HS_E(\lambda)$ as well for any polynomial $P$ with integer coefficients. We can find more elements of $\hs_E(\lambda)$ exponentiating to $\mathbbm{1}$ by conjugating $2\pi P(X_R)$ by an element of $SL(2,\mathbbm{C})$. We find that elements of the form $2\pi P(X^g_R)$, where
\be
X^g_R = -(bd+ac)X_0 + (a^2+b^2)X_+ - (d^2+c^2)X_- \;, \quad ad-bc = 1\;, \quad a,b,c,d \in \mathbbm{C} \;, 
\ee
exponentiate to $\mathbbm{1}$. In particular, picking $b = c = 0$, $d^2 = 1/a^2 = t$, we obtain elements
\be
X^t_R = \frac{1}{2}\left( tV^2_1 + \frac{1}{t} V^2_{-1}\right) \;.
\ee
$X^t_R$ for $t = 2$ coincides exactly with the connection of pure $AdS_3$ integrated along the contractible spatial circle. On the other hand, the connection of BTZ black hole solution, integrated along the contractible time-like circle, is recovered after the identification $t = 2\tau$, where $\tau$ is the (complex) modular parameter of the boundary of a solid torus. (Compare with the first equation in (3.1) of \cite{Kraus:2012uf}, bearing in mind that (3.1) picks up a factor of $\tau$ when integrated.)

To obtain even more solutions, one can conjugate a generic element $2\pi P(X_R)$ with any element of $HS(\lambda)$ (or of $HS_E(\lambda)$). In particular, conjugation with elements of the subgroups $N_+$ and $N_-$ is tractable. As explained in Section \ref{SecClassExpElAdjRep}, their adjoint actions yield finite linear combinations of generators of $\hs_E(\lambda)$.

These results should allow for the construction of new spherically symmetric solutions of the higher spin theory, but we leave this task for future work.

\subsection*{Acknowledgments}

It is a pleasure to thank Anton Alekseev, Alberto Cattaneo and Matthias Gaberdiel for useful discussions, and special thanks go to Matthias Gaberdiel for valuable comments on a draft. This research was supported in part by SNF Grant No.200020-149150/1.

\appendix

\section{Review of relevant mathematical concepts and results}

\label{SecRevMathRes}

The aim of this appendix is to review mathematical results relevant to Sections \ref{SecExpCompSeries} and \ref{SecExFinHSTrans}. General references for the material below include Chapter VIII of \cite{reed1980methods}, \cite{weidmann1980linear}, as well as Section 10.1 and 10.2 of \cite{schmudgen1990unbounded}.

\subsection{Unbounded operators in Hilbert spaces}

\label{SecUnOpHilbSp}

Let $\mathcal{H}$ be a separable Hilbert space with scalar product written $\langle \bullet,\bullet \rangle$. We write $|\bullet|$ for the associated norm. A linear operator on $\mathcal{H}$ is a linear map $T: \mathcal{D}(T) \rightarrow \mathcal{H}$, where the linear subspace $\mathcal{D}(T) \subseteq \mathcal{H}$ is the domain of $T$. Unlike in finite dimension, a linear operator is not necessarily continuous in the topology associated with the norm. In fact, a linear operator is continuous if and only if it is bounded, i.e. if the norm of $T$, defined by
\be
\label{EqDefNormOp}
|T| := \sup_{x \in \mathcal{H}} \frac{|Tx|}{|x|} 
\ee
exists in $\mathbbm{R}$. As $\mathcal{H}$ is complete, we also see that the domain of an unbounded operator cannot coincide with $\mathcal{H}$. Many operators of interest, including the differential operators representing the generators of ${\rm sl}(2,\mathbbm{R})$ in representations of the complementary series, are unbounded operators.

Given operators $T_1$, $T_2$, we say that $T_2$ extends $T_1$, written $T_1 \subset T_2$, if $\mathcal{D}(T_1) \subseteq \mathcal{D}(T_2)$ and $T_2|_{\mathcal{D}(T_1)} = T_1$. We will always assume that the domains of the operators of interest are dense in $\mathcal{H}$.

The graph $\mathcal{G}(T)$ of an operator $T$ is the subset of $\mathcal{H} \times \mathcal{H}$ composed of the set of pairs $(f,Tf)$ for $f \in \mathcal{D}(T)$. $\mathcal{H} \times \mathcal{H}$ carries a scalar product, hence a topology, and we call $T$ a \textit{closed} operator if $\mathcal{G}(T)$ is a closed subset of $\mathcal{H} \times \mathcal{H}$. $T$ is \textit{closable} if it has a closed extension. The smallest closed extension is the \textit{closure} $\bar{T}$ of $T$. A \textit{core} $\mathcal{C}$ of a closable operator $T$ is a subspace $\mathcal{C} \subset \mathcal{D}(T)$ such that $T$ and its restriction to $\mathcal{C}$ have the same closure.

\subsection{Skew-symmetric and skew-adjoint operators}

\label{SecSkSymSkAdjOp}

The \textit{adjoint} $T^\ast$ of an operator $T$ can be defined as follows. Let $\mathcal{D}(T^\ast)$ be the space of $f \in \mathcal{H}$ such that there is $h$ satisfying $\langle Tg,f \rangle = \langle g, h \rangle$ for all $g \in \mathcal{H}$. Then $T^\ast f := h$.

We focus here on skew-symmetric and skew-adjoint operators, but all the results below can be translated for symmetric and self-adjoint operators, using the fact that if $T$ is skew-symmetric/skew-adjoint, then $iT$ is symmetric/self-adjoint. 

An operator is called \textit{skew-symmetric} if $T \subset -T^\ast$. Skew-symmetric operators are always closable. $T$ is called \textit{skew-adjoint} if $T = -T^\ast$, where this equality also implies the equality of the corresponding domains. $T$ is called essentially skew-adjoint if $\bar{T}$ is  skew-adjoint.

It is easy to check that if $T$ is skew-symmetric, then all its eigenvalues are purely imaginary. But its adjoint is not necessarily skew-symmetric, and this provides a criterion to test skew-adjointness. Given a skew-symmetric operator $T$, define its \textit{deficiency indices} $d_\pm = {\rm dim}\, {\rm Ker}(T^\ast \pm 1)$. Then $T$ is essentially skew-adjoint if and only if $d_\pm = 0$. Moreover, $T$ admits a skew-adjoint extension if and only if its deficiency indices are equal. For an example of a skew-symmetric operator with non-zero deficiency indices, see Section \ref{SecV31IsNotEssSA}.

There exists a useful criterion ensuring that the deficiency indices are equal, and therefore that $T$ admits skew-adjoint extensions. Let $K$ be a \textit{conjugation} of $\mathcal{H}$, i.e. an antilinear unitary involution on $\mathcal{H}$. Then $T$ is said to be $K$\textit{-real} if it commutes with $K$ and its domain is preserved by $K$. If $T$ is $K$-real for some conjugation $K$, then it admits skew-adjoint extensions.

\subsection{The functional calculus and exponentiation}

\label{SecFuncCalcExp}

Our interest in skew-adjoint operators stems from the fact that the spectral theorem guarantees that they admit a functional calculus. Informally, given a skew-adjoint operator $T$, the functional calculus maps each bounded complex-valued function $f$ on $\mathbbm{R}$ to a bounded operator $f(iT)$ on $\mathcal{H}$ (see Theorem VIII.5 of \cite{reed1980methods} for a precise formulation). 

In particular, the exponential $U(t) := \exp tT$ is well-defined for all $t \in \mathbbm{R}$ and defines a one-parameter group of unitary operators. The family $U(t)$ is \textit{strongly continuous}, i.e. $\lim_{t \rightarrow t_0} U(t)f = U(t_0)f$ for all $f \in \mathcal{H}$. The limit
\be
\lim_{t \rightarrow 0} \frac{U(t)f - f}{t}
\ee
exists if and only if $f \in \mathcal{D}(T)$ and is equal to $Tf$.

\subsection{Representations of enveloping algebras by unbounded operators}

\label{SecApRepEnvAlgUnbOp}

We now turn to the case where $\mathcal{H}$ is a unitary representation of a Lie group. Let $G$ be a real Lie group and $\rho$ be a unitary representation of $G$ on $\mathcal{H}$. $f \in \mathcal{H}$ is called a $C^\infty$\textit{-vector} of $\rho$ if the map $g \rightarrow \rho(g)f$ from $G$ into $\mathcal{H}$ is $C^\infty$. The set of $C^\infty$-vectors forms a linear subspace $\mathcal{D}^\infty(\rho)$ of $\mathcal{H}$.

Define the operator
\be
d\rho(X)f = \left.\frac{d}{dt} \exp (tX)f  \right|_{t = 0}
\ee
with domain $\mathcal{D}^\infty(\rho)$ for all $X \in \mathfrak{g} := {\rm Lie}(G)$. Then $d\rho$ is a Lie algebra representation of $\mathfrak{g}$ on $\mathcal{H}$ by essentially skew-adjoint operators. It extends to an associative algebra representation of the enveloping algebra $U(\mathfrak{g}_\mathbbm{C})$ of the complexification of $\mathfrak{g}$. $\mathcal{D}^\infty(\rho)$ is a core for $d\rho(X)$ for all $X \in U(\mathfrak{g}_\mathbbm{C})$. 

Let $\theta$ be the anti-automorphism of  $U(\mathfrak{g}_\mathbbm{C})$ acting by $-1$ on $\mathfrak{g} \subset U(\mathfrak{g}_\mathbbm{C})$, and $U^\theta(\mathfrak{g})$ the eigenspace of eigenvalue $-1$. We call the elements of $U^\theta(\mathfrak{g})$ \textit{skew-hermitian elements}. They are represented on $\mathcal{H}$ by skew-symmetric operators. Remark that in the case of interest to us, the quotient of $U^\theta({\rm sl}(2,\mathbbm{R}))$ by the ideal generated by $\Omega - \frac{1}{2}(\lambda^2-1) \mathbbm{1}$ coincides with $\hs(\lambda)$, so we recover the fact that $\hs(\lambda)$ is represented by skew-symmetric operators on $\mathcal{H}_\lambda$ in the complementary series.

\subsection{Criterion for skew-adjointness}

\label{SecCritSAEnvAlg}

What remains to be understood is under which conditions a skew-hermitian element of the enveloping algebra is represented by a skew-adjoint operator (and therefore exponentiates). We present below a sufficient criterion due to Nelson and Stinespring \cite{Nelson1959} (see also Section 10.2 of \cite{schmudgen1990unbounded}). 

To this end, we need to introduce the notion of elliptic elements of the enveloping algebra. Recall that we can associate to any generator of $\mathfrak{g}$ a left-invariant vector field on the group manifold $G$. Therefore, we can associate a differential operator on the space of complex-valued smooth functions on $G$ to each element $X$ of $U(\mathfrak{g}_{\mathbbm{C}})$. An element of $U(\mathfrak{g}_{\mathbbm{C}})$ is said to be \textit{elliptic} if the corresponding differential operator is elliptic. Practically, we can pick a basis $\{X_i\}$ of $\mathfrak{g}$ and use the Poincaré-Birkhoff-Witt theorem to write
\be
X = \sum_{[n]} \alpha_{[n]} X^{[n]} \;,
\ee
where $[n] = (n_1,...,n_d)$, $\alpha_{[n]} \in \mathbbm{C}$ and $X^n = X_1^{n_1}...X_d^{n_d}$. We define $|X|$ to be the maximal value of $\sum_i n_i$ such that $\alpha_{[n]} \neq 0$. $|X|$ is the degree of the associated differential operator on $G$. The symbol of $X$ is
\be
\sigma_X: \mathbbm{R}^d \rightarrow \mathbbm{C}\;, \quad \sigma_X(x_1,...,x_d) = \sum_{\stackrel{[n]}{n_1 +...+ n_d = |X|}} \alpha_{[n]} x_1^{n_1} ... x_d^{n_d} \;.
\ee
Now $X$ is elliptic if $X$ is not a multiple of the identity and if $\sigma_X(x_1,...,x_d) = 0$ implies $x_1 = ... = x_d = 0$. Remark that as the ellipticity depends only on the highest degree component of $X$, it is always possible to add an arbitrary small higher degree term to $X$ in order to make it elliptic. In this sense, the set of elliptic elements is dense in $U(\mathfrak{g}_{\mathbbm{C}})$.

Nelson and Stinespring proved \cite{Nelson1959}:
\begin{theorem}
For any skew-hermitian elliptic element $X \in U(\mathfrak{g}_\mathbbm{C})$, $d\rho(X)$ is essentially skew-adjoint.
\end{theorem}
\begin{theorem}
\label{ThComElSAEl}
Let $X$ be a skew-hermitian element of $U(\mathfrak{g}_\mathbbm{C})$ and $Y$ an elliptic element such that $d\rho(X)$ commutes with $d\rho(Y)$. Then $d\rho(X)$ is essentially skew-adjoint.
\end{theorem}
The following result will also be useful to us \cite{schmudgen1990unbounded}.
\begin{theorem}
Let $X \in \mathfrak{g}$ and let $p$ be a polynomial with real coefficients. Then $d\rho(ip(iX))$ is essentially skew-adjoint.
\end{theorem}

Remark that for a compact group, the quadratic Casimir is an elliptic element. Theorem \ref{ThComElSAEl} then implies that all the skew-hermitian elements of the enveloping algebra are represented by essentially skew-adjoint operators. For $SL(2,\mathbbm{R})$, the case of interest to us, the quadratic Casimir is not elliptic and we exhibit in Section \ref{SecV31IsNotEssSA} a skew-hermitian element represented by a non-essentially skew-adjoint operator.

\section{Elements of $\hs(\lambda)$ represented by operators admitting skew-adjoint extensions}

\label{SecExtNonEssSAOp}

In the main body of the paper, we considered an element of $\hs(\lambda)$ to be exponentiable if and only if it is represented in the complementary series of representations by an essentially skew-adjoint operator. This definition might seem too strong, as there are elements of $\hs(\lambda)$ represented by operators admitting non-unique skew-adjoint extensions. We mention here a method to characterize such operators, and remark that all the elements of the orthogonal subalgebra of $\hs(\lambda)$ admit skew-adjoint extensions.

In order to make sense of their exponentials, non essentially skew-adjoint operators need to be extended to skew-adjoint operators. Skew-adjoint extensions do not always exist, and when they do, they are not unique. As we mentioned in Appendix  \ref{SecSkSymSkAdjOp}, a skew-symmetric operator is essentially skew-adjoint if and only if its deficiency indices are both equal to zero, and admits skew-adjoint completions if and only if its deficiency indices are equal. In the case of ordinary differential operators, there is no efficient way of computing the deficiency indices (see for instance \cite{everitt1999boundary} on p.86). However there is a nice criterion ensuring that an operator admits skew-adjoint extension. An operator which is real with respect to a conjugation of $\mathcal{H}$ necessarily admits skew-adjoint extensions. The existence of the complex conjugation on $\mathcal{H}$ then ensures that the operators $V^s_n$ for $s$ even, as well as their linear combinations in $\hs(\lambda)$, admit skew-adjoint extensions.

Interestingly, the even spin generators span a subalgebra, the orthogonal higher spin Lie algebra $\hs_O(\lambda)$. The latter can be used to build a higher spin theory containing one even spin field for each spin (see for instance \cite{Ahn:2011pv, Gaberdiel:2011nt}). The argument above shows that all the elements of $\hs_O(\lambda)$ are represented by operators admitting skew-adjoint extensions. If there was a way of singling out a particular completion for each element of $\hs_O(\lambda)$, they would all exponentiate.

{
\small

\providecommand{\href}[2]{#2}\begingroup\raggedright\endgroup

}

\end{document}